\documentclass[a4paper,UKenglish,cleveref, autoref, thm-restate]{lipics-v2019}

\usepackage{amsmath}
\usepackage{amsthm}
\usepackage{amssymb}
\usepackage{amsfonts}
\usepackage{color}
\usepackage{comment}
\usepackage{algorithm}
\usepackage{algorithmicx}
\usepackage{caption}
\usepackage{ragged2e}
\usepackage{calc}
\usepackage{epic}
\usepackage{eepic}
\usepackage{epsfig,float}
\usepackage{verbatim}
\usepackage{graphicx}
\usepackage{mathtools}
\usepackage{amsmath}
\usepackage[noend]{algpseudocode}
\usepackage{fancyhdr}
\usepackage{bbm}
\usepackage{xcolor}

\DeclareMathOperator{\down}{\downarrow}
\DeclareMathOperator{\up}{\uparrow}

\newcommand{\res}{\text{res}}
\newcommand{\noin}{\noindent}
\newcommand{\be}{\begin{enumerate}}
\newcommand{\ee}{\end{enumerate}}
\newcommand{\bi}{\begin{itemize}}
\newcommand{\ei}{\end{itemize}}

\renewcommand{\le}{\leqslant}
\renewcommand{\ge}{\geqslant}
\newcommand{\meet}{\land}
\newcommand{\join}{\lor}

\newcommand{\DOWN}{\text{DOWN}}
\newcommand{\NULL}{\texttt{null}}

\newcommand{\cA}{\mathcal{A}}
\newcommand{\cB}{\mathcal{B}}
\newcommand{\cC}{\mathcal{C}}
\newcommand{\cD}{\mathcal{D}}
\newcommand{\cE}{\mathcal{E}}

\EventEditors{Susanne Albers}
\EventNoEds{1}
\EventLongTitle{17th Scandinavian Symposium and Workshops on Algorithm Theory (SWAT 2020)}
\EventShortTitle{SWAT 2020}
\EventAcronym{SWAT}
\EventYear{2020}
\EventDate{June 22--24, 2020}
\EventLocation{T\'{o}rshavn, Faroe Islands}
\EventLogo{}
\SeriesVolume{162}
\ArticleNo{31}

\title{Space-Efficient Data Structures for Lattices}

\author{J. Ian Munro}{Cheriton School of Computer Science, University of Waterloo, Waterloo, Ontario N2L 3G1, Canada}{imunro@uwaterloo.ca}{https://orcid.org/0000-0002-7165-7988}{}

\author{Bryce Sandlund}{Cheriton School of Computer Science, University of Waterloo, Waterloo, Ontario N2L 3G1, Canada}{bcsandlund@uwaterloo.ca}{}{}

\author{Corwin Sinnamon}{Department of Computer Science, Princeton University, Princeton, NJ 08540, USA}{sinncore@gmail.com}{https://orcid.org/0000-0003-0280-9498}{}

\authorrunning{J.\,I. Munro, B. Sandlund, and C. Sinnamon}

\Copyright{James Ian Munro, Bryce Sandlund, and Corwin Sinnamon}



\category{} 

\supplement{}

\funding{This work was supported by the Natural Sciences and Engineering Research Council of Canada and the Canada Research Chairs Program.}

\nolinenumbers 

\date{}

\begin{document}

\maketitle

\begin{abstract}
A lattice is a partially-ordered set in which every pair of elements has a unique meet (greatest lower bound) and join (least upper bound).
We present new data structures for lattices that are simple, efficient, and nearly optimal in terms of space complexity.

Our first data structure can answer partial order queries in constant time and find the meet or join of two elements in $O(n^{3/4})$ time, where $n$ is the number of elements in the lattice.
It occupies $O(n^{3/2}\log n)$ bits of space, which is only a $\Theta(\log n)$ factor from the $\Theta(n^{3/2})$-bit lower bound for storing lattices.
The preprocessing time is $O(n^2)$.
This structure admits a simple space-time tradeoff so that, for any $c \in [\frac{1}{2}, 1]$, the data structure supports meet and join queries in $O(n^{1-c/2})$ time, occupies $O(n^{1+c}\log n)$ bits of space, and can be constructed in $O(n^2 + n^{1+3c/2})$ time.

Our second data structure uses $O(n^{3/2}\log n)$ bits of space and supports meet and join in $O(d \frac{\log n}{\log d})$ time, where $d$ is the maximum degree of any element in the transitive reduction graph of the lattice. This structure is much faster for lattices with low-degree elements.

This paper also identifies an error in a long-standing solution to the problem of representing lattices.
We discuss the issue with this previous work.
\end{abstract}

\section{Introduction}
A lattice is a partially-ordered set with the property that for any pair of elements $x$ and $y$, the set of all elements greater than or equal to both $x$ and $y$ must contain a unique minimal element less than all others in the set. This element is called the \emph{join} (or \emph{least upper bound}) of $x$ and $y$. A similar condition holds for the set of all elements less than both $x$ and $y$: It must contain a maximum element called the \emph{meet} (or \emph{greatest lower bound}) of $x$ and $y$.

We consider lattices from the perspective of succinct data structures.
This area of study is concerned with representing a combinatorial object in essentially the minimum number of bits while supporting the ``natural'' operations in constant time.
The minimum number of bits required is the logarithm (base 2) of the number of such objects of size $n$, e.g. about 2n bits for a binary tree on $n$ nodes.
Succinct data structures have been very successful in dealing with trees, planar graphs, and arbitrary graphs.
Our goal in this paper is to broaden the horizon for succinct and space-efficient data structures and to move to more algebraic structures.
There has indeed been progress in this direction with abelian groups~\cite{FM06} and distributive lattices~\cite{MS18}.
We take another step here in studying space-efficient data structures for arbitrary finite lattices.

Lattices have a long and rich history spanning many disciplines.
Existing at the intersection of order theory and abstract algebra, lattices arise naturally in virtually every area of mathematics~\cite{Gra16}.
The area of formal concept analysis is based on the notion of a \emph{concept lattice}. These lattices have been studied since the 1980s~\cite{Wil82} and have applications in linguistics, data mining, and knowledge management, among many others~\cite{Gan05}.
Lattices have also found numerous applications in the social sciences~\cite{Mon03}.

Within computer science, lattices are also important, particularly for programming languages.
Lattice theory is the basis for many techniques in static analysis of programs, and thus has applications to compiler design.
Dataflow analysis and abstract interpretation, two major areas of static analysis, rely on fixed-point computations on lattices to draw conclusions about the behaviour of a program~\cite{Nie15}.

Lattice operations appear in the problem of hierarchical encoding, which is relevant to implementing type inclusion for programming languages with multiple inheritance (among other applications)~\cite{ABLN89, Cas93, Cas99, Kra97}.
Here the problem is to represent a partially-ordered set by assigning a short binary string to each element so that lattice-like operations can be implemented using bitwise operations on these strings.
The goal is to minimize the length of the strings for the sake of time and space efficiency.


In short, lattices are pervasive and worthy of study.
From a data structures perspective, the natural question follows: How do we represent a lattice so that not too much space is required and basic operations like partial order testing, meet, and join can be performed quickly?

It was proven by Klotz and Lucht~\cite{KL71} that the number of different lattices on $n$ elements is at least $2^{\Omega(n^{3/2})}$,
and an upper bound of $2^{O(n^{3/2})}$ was shown by Kleitman and Winston~\cite{KW80}.
Thus, any representation for lattices must use $\Omega(n^{3/2})$ bits in the worst case, and this lower bound is tight within a constant factor.
We should then expect a data structure for lattices to use comparably little space.

Two naive solutions suggest themselves immediately. 
First, we could simply build a table containing the meet and join of every pair of elements in the given lattice. Any simple lattice operation could be performed in constant time. However, the space usage would be quadratic --- a good deal larger than the lower bound.
Alternatively, we could store only the transitive reduction graph of the lattice. This method turns out to be quite space-efficient: Since the transitive reduction graph of a lattice can only have $O(n^{3/2})$ edges~\cite{KL71, Win79}, the graph can be stored in $O(n^{3/2}\log n)$ bits of space; thus, the space complexity lies within a $\Theta(\log n)$ factor of the lower bound. However, the lattice operations become extremely slow as they require exhaustively searching through the graph.
Indeed, it is not easy to come up with a data structure for lattices that uses less than quadratic space while answering meet, join, and partial order queries in less than linear time in the worst case.


The construction of a lattice data structure with good worst-case behaviour also has attractive connections to the more general problem of reachability in directed acyclic graphs (DAGs). Through its transitive reduction graph, a lattice can be viewed as a special type of DAG.
Among other things, this paper shows that we can support reachability queries in constant time for this class of graphs while using subquadratic space.
Most classes of DAGs for which this has been achieved, such as planar DAGs~\cite{Tho04}, permit a strong bound on the order dimension of the DAGs within that class.
This is a property not shared by lattices, which may have order dimension linear in the size of the lattice.
A long-standing difficult problem in this line of research is to show a similar nontrivial result for the case of arbitrary sparse DAGs~\cite{Pat11}.

There has been significant progress in representation of \emph{distributive} lattices, an especially common and important class of lattices.
Space-efficient data structures for distributive lattices have been established since the 1990s~\cite{Hab01, Hab96} and have been studied most recently by Munro and Sinnamon~\cite{MS18}.
Munro and Sinnamon show that it is possible to represent a distributive lattice on $n$ elements using $O(n\log n)$ bits of space while supporting meet and join operations (and thus partial order testing) in $O(\log n)$ time.
This comes within a $\Theta(\log n)$ factor of the space lower bound by enumeration: As the number of distributive lattices on $n$ elements is $2^{\Theta(n)}$~\cite{EHR02}, at least $\Theta(n)$ bits of space are required for any representation.

The problem of developing a space-efficient data structure for arbitrary lattices was first studied by Talamo and Vocca in 1994, 1997, and 1999~\cite{TV94,TV97,TV99}.
They claimed to have an $O(n^{3/2}\log n)$-bit data structure that supports partial order queries in constant time and meet and join operations in $O(\sqrt{n})$ time.
However, there is a nontrivial error in the details of their structure.
Although much of the data structure is correct, we believe that this mistake is a critical flaw that is not easily repaired.

To our knowledge, no other data structures have been proposed that can perform lattice operations efficiently while using less than quadratic space.
Our primary motivation is to fill this gap.

\section{Contributions}
\label{sec:contributions}
Drawing on ideas from~\cite{TV99}, we present new data structures for lattices that are simple, efficient for the natural lattice operations, and nearly optimal in space complexity.
Our data structures support three queries:
\bi
\item \textbf{Test Order:} Given two elements $x$ and $y$, determine whether $x \le y$ in the lattice order.
\item \textbf{Find Meet:} Find the meet of two elements.
\item \textbf{Find Join:} Find the join of two elements.
\ei
Our first data structure (Theorem~\ref{thm:main}) is based on a two-level decomposition of a lattice into many smaller lattices.
It tests the order between any two elements in $O(1)$ time and answers meet and join queries in $O(n^{3/4})$ time in the worst case.
It uses $O(n^{3/2})$ words of space\footnote{We assume a word RAM model with $\Theta(\log n)$-bit words. Henceforth, unless bits are specified, ``$f(n)$ space'' means $f(n)$ words of size $\Theta(\log n)$.}, which is a $\Theta(\log n)$ factor from the known lower bound of $\Omega(n^{3/2})$ bits.
The preprocessing time is $O(n^2)$.

We generalize this structure (Corollary~\ref{cor:generalize}) to allow for a tradeoff between the time and space requirements. For any $c \in [\frac{1}{2}, 1]$, we give a data structure that supports meet and join operations in $O(n^{1-c/2})$ time, occupies $O(n^{1+c})$ space, and can be constructed in $O(n^2 + n^{1+3c/2})$ time. At $c=1/2$, it coincides with the first data structure.

Taking a different approach to computing meets and joins, we present another data structure (Theorem~\ref{thm:recurse}) based on a recursive decomposition of the lattice. Here the operational complexity is parameterized by the maximum degree $d$ of any element in the lattice, where the degree is defined in reference to the transitive reduction graph of the lattice. This structure answers meet and join queries in $O(d \frac{\log n}{\log d})$ time, which improves significantly on the first data structure when applied to lattices with low degree elements (as is the case for distributive lattices, for example). It uses $O(n^{3/2})$ space.

This paper is organized as follows. In Section~\ref{sec:prelim}, we give the necessary definitions and notation used throughout the paper. In Section~\ref{sec:block}, we give the main tool we use to decompose a lattice, which we call a block decomposition. Section~\ref{sec:order} describes the order-testing data structure and Section~\ref{sec:meet} extends this data structure to compute meets and joins. Some details of the preprocessing are left to Appendix~\ref{sec:initialize}.
Section~\ref{sec:deg} contains our recursive degree-bounded data structure.
In Appendix~\ref{sec:error}, we discuss the error in the papers~\cite{TV97, TV99} and give some evidence of why it may be irreparable.

\section{Preliminaries}
\label{sec:prelim}
Given a partially-ordered set (poset) $(P, \le)$, we define the \emph{downset} of an element $x \in P$ by $\down x = \{ z \in P \mid z \le x\}$ and the \emph{upset} of $x$ by $\up x = \{ z \in P \mid z \ge x\}$.
\begin{definition}
A \emph{lattice} is a partially-ordered set $(L, \le)$ in which every pair of elements has a \emph{meet} and a \emph{join}.

The \emph{meet} of $x$ and $y$, denoted $x \meet y$, is the unique maximal element of $\down x \cap \down y$ with respect to $\le$.
Similarly, the \emph{join} of $x$ and $y$, denoted $x \join y$, is the unique minimal element of $\up x \cap \up y$.
\end{definition}
Meet ($\meet$) and join ($\join$) are also called greatest lower bound (GLB) and least upper bound (LUB), respectively.
Lattices have the following elementary properties. Let $x, y, z \in L$.
\bi
\item
The meet and join operations are idempotent, associative, and commutative:
\begin{align*}
x \join x &= x & x \join (y \join z) &= (x \join y) \join z& x \join y &= y \join x\\
x \meet x &= x & x \meet (y \meet z) &= (x \meet y) \meet z& x \meet y &= y \meet x
\end{align*}
\item
If $x \le y$, then $x \meet y = x$ and $x \join y = y$.
\item
If $z \le x$ and $z \le y$, then $z \le x \meet y$. If $z \ge x$ and $z \ge y$, then $z \ge x \join y$.
\item
A lattice must have a unique \emph{top} element above all others and unique \emph{bottom} element below all others in the lattice order.
\ei
Moreover, meet and join are dual operations.
If the lattice is flipped upside-down, then meet become join and vice versa.

In this paper, we prefer to work with \emph{partial lattices}.
A partial lattice is the same as a lattice except that it does not necessarily have top or bottom elements.
Thus, the meet or join of two elements may not exist in a partial lattice; we use the symbol $\NULL$ to indicate this.
We write $x \meet y = \NULL$ if $\down x \cap \down y = \emptyset$ and $x \join y = \NULL$ if $\up x \cap \up y = \emptyset$.
Note that in a partial lattice the meet or join of $x$ and $y$ may not exist, but when they do exist they must be unique.

Equivalently, a partial lattice is a partially-ordered set satisfying the \emph{lattice property}:
If there are four elements $x_1$, $x_2$, $y_1$, and $y_2$ such that $x_1, x_2 < y_1, y_2$, then there must exists an intermediate element $z$ with $x_1, x_2 \le z \le y_1, y_2$.
See Figure~\ref{fig:butterfly}.
This statement trivially follows from the definition of a lattice; it only says that there cannot be multiple maximal elements in $\down y_1 \cap \down y_2$ or multiple minimal elements in $\up x_1 \cap \up x_2$.

Henceforth, we use the term ``lattice'' to mean ``partial lattice''.
The difference is trivial in a practical sense, and our results are easier to express when we only consider partial lattices.

\begin{figure}[!ht]
\begin{center}
\includegraphics[scale=0.7]{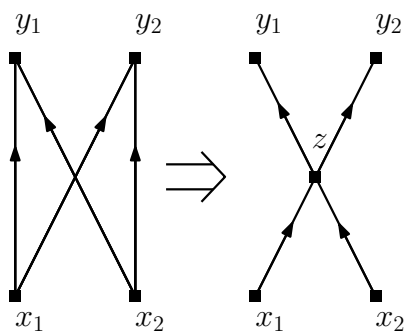}
\end{center}
\caption{The configuration on the left cannot exist in a lattice for any nodes $x_1$, $x_2$, $y_1$, and $y_2$. There must be a node $z$ between them as shown. We refer to this as the \emph{lattice property}.}
\label{fig:butterfly}
\end{figure}

We assume that any lattice we wish to represent is given initially its \emph{transitive reduction graph} (TRG).
This is a directed acyclic graph (DAG) having a node for each lattice element and an edge $(u, v)$ whenever $u < v$ and there is no intermediate node $w$ such that $u < w < v$.
The edge relation of this graph is called the \emph{covering} relation: Whenever $(u,v)$ is an edge of the TRG we say that $v$ \emph{covers} $u$.

\section{Block Decompositions}
\label{sec:block}

The main tool used in our data structure is called a block decomposition of a lattice. It is closely based on techniques used by Talamo and Vocca in~\cite{TV97, TV99}.

Let $L$ be a lattice with $n$ elements.
A block decomposition of $L$ is a partition of the elements of $L$ into subsets called \emph{blocks}.
The blocks are chosen algorithmically using the following method.
We first specify a positive integer $k$ to be the \emph{block size} of the decomposition (our application will use the block size $\sqrt n$).
Then we label the elements of $L$ as ``fat'' or ``thin'' according to the sizes of their downsets. A fat node is ``minimal'' if all elements in its downset, except itself, are thin.
Formally:

\begin{definition}
A node $x \in L$ is called \emph{fat} if $|\down x| \geq k$, and $x$ is called \emph{thin}  if $|\down x| < k$.
We say $x$ is a \emph{minimal fat node} if $x$ is fat and every other node in $\down x$ is thin.
\end{definition}

Minimal fat nodes are the basis for choosing blocks, which is done as follows.
While there exists a minimal fat node $h$ in the lattice, create a new \emph{principal} block $B$ containing the elements of $\down h$, and then delete those nodes from the lattice.
The node $h$ is called the \emph{block header} of $B$.

Deleting the elements of $B$ may cause some fat nodes to become thin by removing elements from their downsets; this should be accounted for before choosing the next block.
When there are no fat nodes in the lattice, put the remaining elements into a single block $B_{\res}$ called the \emph{residual block}.

\begin{figure}[!ht]
\begin{center}
\includegraphics[scale=0.6]{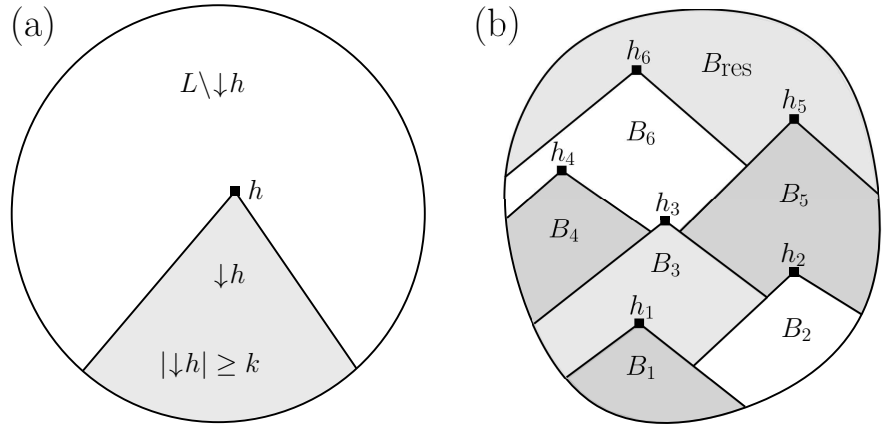}
\end{center}
\caption{(a) A minimal fat node $h$ is used as a block header during the decomposition. The downset of $h$ is removed and the process repeats on $L \setminus \down h$. (b) A block decomposition yields a set of disjoint principal blocks, each having a block header. The residual block consists of the lattice elements that are not below any block header.}
\label{fig:blockdecomp}
\end{figure}

This method creates a set of principal blocks $\{B_1, B_2, \dots, B_m\}$ and a residual block $B_{\res}$.
Each principal block $B_i$ has a block header $h_i$, which was the minimal fat node used to create $B_i$.
A block header is always the top element within its block.
The residual block may or may not have a top element, but it is not considered to have a block header regardless. Figure~\ref{fig:blockdecomp} shows a full block decomposition.

The block decomposition algorithm is summarized in Algorithm~\ref{alg:decomp_intuitive}; it will be shown later that this algorithm can be implemented to run in $O(n^{7/4})$ time, where $n$ is the number of elements in the lattice.

\begin{algorithm}[!ht]
\caption{Block Decomposition (Intuitive Version)}
\label{alg:decomp_intuitive}
\begin{algorithmic}[1]
\Require{A partial lattice $L$ on $n$ elements and a positive integer $k$.}
\Ensure{A block decomposition of $L$ with block size $k$.}
\State $i = 1$
\While{there exists a minimal fat node $h$}
\State $B_i = \down h \cap L$
\State $L = L \setminus B_i$
\State $i = i+1$
\EndWhile
\State $B_\res = L$
\end{algorithmic}
\end{algorithm}

\subsection*{Properties of Block Decompositions}


Let us note some elementary properties of block decompositions. Let $L$ be a lattice with $n$ elements.
\bi
\item Every element of the lattice lies in exactly one block.
\item There can be at most $n/k$ principal blocks as each one has size between $k$ and $n$.
Consequently, there are at most $n/k$ block headers.
\item Since the block headers are chosen to be \emph{minimal} fat nodes, every other element is thin relative to the block it lies in.
That is, if $x$ lies in a block $B$ and $x$ is not the block header of $B$, then $|\down x \cap B| < k$.
\ei

The last fact motivates the following term, which we will use frequently.
\begin{definition}
The \emph{local downset} of an element $x$ is the set $\down x \cap B$, where $B$ is the block containing $x$.
\end{definition}
Restated, the last property listed above says that the local downset of any element that is not a block header has size less than $k$.
We also note that if $h$ is the block header of a principal block $B$, then the local downset of $h$ is $B$.

Somewhat less obvious is the following lemma.

\begin{lemma}
\label{lem:blocksAreLattices}
Every block is a partial lattice.\footnote{Here the partial order on a block is inherited from the order on $L$.}
\end{lemma}
\begin{proof}
The lemma follows from two facts.
\be
\item The downset of any element in a partial lattice is also a partial lattice.
\item If the downset of an element is removed from a partial lattice, then the remaining elements still form a partial lattice.
\ee
We prove the first fact.
Let $h$ be an element of a partial lattice $L$.
We prove that the poset $\down h$ satisfies the lattice property (see Figure~\ref{fig:butterfly}).
Suppose there are four elements $x_1, x_2, y_1, y_2 \in \down h$ such that $x_1, x_2 < y_1, y_2$.
These elements also lie in $L$, and since $L$ is a lattice there must be an element $z \in L$ such that $x_1, x_2 \le z \le y_1, y_2$.
As $z \le y_1 \le h$, $z$ must lie in $\down h$.
Thus $\down h$ is a partial lattice because it satisfies the lattice property.

The second fact is similar.
Suppose $\down h$ is removed from a partial lattice $L$.
If there are four elements $x_1, x_2, y_1, y_2 \in L \setminus \down h$ with $x_1, x_2 < y_1, y_2$, then there must be an element $z \in L$ with $x_1, x_2 \le z \le y_1, y_2$.
This element $z$ cannot lie in $\down h$ because $x_1 \le z$ and $x_1 \not\in \down h$.
Therefore $z \in L \setminus \down h$.
\end{proof}

\begin{remark}
\emph{To avoid confusion in our notation, all lattice relations and operators are assumed to be with respect to $L$.
In particular, $\meet$, $\join$, $\up$, and $\down$ always reference the full lattice and are not restricted to a single block.}
\end{remark}

\subsection*{Intuition for Block Decompositions}

We can now explain intuitively why a block decomposition is a good idea and how it leads to an effective data structure.
Lemma~\ref{lem:blocksAreLattices} means that the blocks can be treated as independent partial lattices.
Moreover, the elements within each block are all thin, with the noteworthy exception of the block headers.
For any single block, this thinness condition makes it possible to create a fast, simple, space-efficient data structure that facilitates computations within that block.
However, such a data structure only contains local information about its block; it cannot handle operations that span multiple blocks.

For those operations, we rely upon the block headers to bridge the gaps.
The block headers are significant because they induce a \emph{unique representative property} on the blocks:
If $h$ is the block header of some principal block $B$ and $x$ is some element of the lattice, then we think of $x \meet h$ as the representative of $x$ in block $B$.
For all of the operations that we care about, the representative of $x$ in $B$ faithfully serves the role of $x$ during computations within $B$.

Combining the power of the unique representative property with our ability to quickly perform block-local operations gives us an effective data structure for lattices, which we are now prepared to describe.

\section{A Data Structure for Order Testing}
\label{sec:order}
First, we describe a simple data structure that performs order-testing queries (answers ``Is $x \le y$?'') in constant time.
We later extend it to handle meet and join queries as well.

Given a partial lattice $L$ with $n$ elements, we perform a block decomposition on $L$ using the block size $k=\sqrt n$.
Let $B_1, B_2, \dots, B_m$, and $B_\res$ be the blocks of this decomposition and $h_1, \dots, h_m$ be the block headers.
Note that $m \leq \sqrt n$.

\subsection*{Information Stored}

We represent each element of $L$ by a node with two fields.\footnote{We often use the term ``node'' to refer to the element of $L$ that the node represents.}
One field contains a unique \emph{identifier} for the lattice element, a number between $0$ and $n-1$, for indexing purposes.
The other field indicates the block that the element belongs to.

Our data structure consists of ($\cA$), a collection of arrays, and ($\cB$), a collection of dictionaries.

\be
\item[($\cA$)] For each block header $h_i$, we store an array containing a pointer to the node $h_i \meet x$ for each $x \in L$. The meet of any node with any block header can be found with one access to the appropriate array.

\item[($\cB$)] For each $x \in L$ we store a dictionary $\DOWN(x)$ containing the identifiers of all the nodes in the local downset of $x$. By using a space-efficient static dictionary (e.g.~\cite{BM99}), membership queries can be performed in constant time.
With this, we can test the order between any two nodes in the same block in constant time.
\ee



\subsection*{Testing Whether $x \le y$}

Given nodes $x$ and $y$ in $L$, we can test whether $x \le y$ in three cases.
\be[{Case} 1:]
\item
If $x$ is in a principal block $B_i$, then find $y_i = h_i \meet y$ using ($\cA$).
If $y_i \in B_i$, then $x \le y$ if and only if $x$ is a member of $\DOWN(y_i)$; this can be tested using ($\cB$).
If $y_i \not\in B_i$, then $x \not\le y$.
\item
If $x \in B_\res$ and $y \in B_\res$, then $x \le y$ if and only if $x$ is a member of $\DOWN(y)$.
\item
If $x \in B_\res$ and $y \not\in B_\res$, then $x \not\le y$.
\ee

The three cases can be tested in constant time using ($\cA$) and ($\cB$).

\begin{proposition}
The above method correctly answers order queries.
\end{proposition}
\begin{proof}
Clearly the three cases cover all possibilities for $x$ and $y$.

In Case 1, $y_i = h_i \meet y$ has the property that $x \le y$ if and only if $x \le y_i$.
This property holds because $x \le h_i$ by assumption, and by the definition of meet, \[x \le h_i \meet y \text{ if and only if } x \le h_i \text{ and } x \le y.\]
If $y_i \in B_i$, then the order can be tested directly using $\DOWN(y_i)$.
If $y_i \not\in B_i$, then $y_i$ cannot be above $x$ in the lattice because $y_i \le h_i$ and every element between $x$ and $h_i$ must lie in $B_i$.

Case 2 is checked directly using ($\cB$).

Case 3 is correct because $B_\res$ consists of all elements that are not below any block header.
As $y$ is in some principal block, it must lie below some block header. Hence, $x$ cannot be below $y$.
\end{proof}

\subsection*{Space Complexity}

Storing the $n$ nodes of the lattice requires $\Theta(n)$ space.
Each array of ($\cA$) requires $\Theta(n)$ space and there are at most $\sqrt{n}$ block headers, yielding $O(n^{3/2})$ space in total.

Assuming ($\cB$) uses a succinct static dictionary (see~\cite{BM99}), the space usage for ($\cB$) will be proportional to the sum of $|\down x \cap B_x|$ over all $x \in L$, where $B_x$ is the block containing $x$.
If $x$ is not a block header, then $|\down x \cap B_x| < \sqrt n$ because the local downsets must be smaller than the block size of the decomposition.
There are $n-m$ such elements, as $m$ denotes the number of principal blocks.
If $x$ is a block header, then $\down x \cap B_x = B_x$.
Thus
\[\sum_{x \in L}|\down x \cap B_x| \leq (n-m)\sqrt n + \sum_{i=1}^m |B_i| \leq (n-m)\sqrt n + n \leq 2n^{3/2}.\]

The total space for the data structure is therefore $O(n^{3/2})$.

\section{Finding Meets and Joins}
\label{sec:meet}

We now extend the order-testing data structure of the last section to answer meet queries: Given two elements $x$ and $y$ in $L$, we wish to find $x \meet y$. Our data structure can answer these queries in $O(n^{3/4})$ time.


\subsection*{Subblock Decompositions}

Let $B_i$ be a principal block with block header $h_i$.
A \emph{subblock decomposition} of $B_i$ is simply a block decomposition of $B_i \setminus \{h_i\}$.

To state it explicitly, the subblock decomposition is a partition of $B_i \setminus {h_i}$ into a set of principal subblocks $\{S_{i,1}, S_{i, 2}, \dots, S_{i,\ell_i}\}$, each having a subblock header $g_{i,j}$, and one residual subblock $S_{i,\res}$.
The decomposition strategy is identical to that of a block decomposition, and it still depends on a \emph{subblock size $r$} that we specify.

We exclude $h_i$ from the subblock decomposition as a convenience. We want to use the property that the local downsets of the elements in $B_i$ have size less than $\sqrt n$, and this holds for every element of $B_i$ except for $h_i$.

Obviously, the subblocks have the same properties as blocks.
\bi
\item Each principal subblock $S_{i,j}$ is a subset of $B_i$ with $|S_{i,j}| \geq r$. Hence, $\ell_i \leq \frac{|B_i|}{r}$.
\item If $x \in S_{i,j} \setminus \{g_{i,j}\}$ then $|\down x \cap S_{i,j}| < r$.
\item If $x \in S_{i,\res}$ then $|\down x \cap S_{i,\res}| < r$.
\item Each subblock is a partial lattice.
\ei


\subsection*{Extending the Data Structure}

As before, let $B_1, B_2, \dots, B_m$, and $B_\res$ be the blocks of the decomposition of $L$, each having size at least $\sqrt n$.
Within each principal block $B_i$, we perform a subblock decomposition with subblock size $r = \sqrt{|B_i|}$, yielding subblocks $S_{i,1}, S_{i, 2}, \dots, S_{i,\ell_i}$, and $S_{i,\res}$.
We have $\ell_i \leq \sqrt{|B_i|}$ for $1 \leq i \leq m$.
There is a subblock header $g_{i,j}$ for each principal subblock $S_{i,j}$, $1 \leq i \leq m$ and $1 \leq j \leq \ell_i$.

\subsection*{Information Stored}


We add a new field to each node that indicates which subblock contains it.
We store ($\cA$) and ($\cB$) as in the order-testing structure, and additionally:

\bi
\item[($\cC$)] For each subblock header $g_{i, j}$, we store an array containing a pointer to $g_{i,j} \meet x$ for all $x \in B_i$. These arrays allow us to determine the meet of any subblock header and any node in the same block with a single access.

\item[($\cD$)] For each principal subblock $S_{i,j}$, we store a table that contains the meet of each pair of elements from $S_{i,j}$, unless the meet lies outside $S_{i,j}$.
That is, the table has $|S_{i,j}|^2$ entries indexed by pairs of elements in $S_{i,j}$. The entry for $(x, y)$ contains a pointer to $x \meet y$ if it lies in $S_{i,j}$, or $\NULL$ otherwise.
We can compute meets within any principal subblock in constant time using these tables.

\item[($\cE$)] For every element $x$ in a residual subblock $S_{i,\res}$, we store $\down x \cap S_{i,\res}$ as a linked list of pointers. This allows us to iterate through the local downset of each element in the residual subblock.
\ei

\subsection*{Finding the Meet}

This data structure allows us to find the meet of two elements $x, y \in L$ in $O(n^{3/4})$ time.
The meet-finding operation works by finding representative elements for $x$ and $y$ in each principal block and computing the meet of each pair of representatives.
We call these \emph{candidate meets} for $x$ and $y$.
Once the set of candidate meets is compiled, the algorithm finds the largest element among them (with respect to the lattice order) and returns it.

We refer to the algorithm as \textproc{Meet}.
This algorithm uses a subroutine called \textproc{Meet-In-Block} that finds the meet of two elements from the same principal block, or else determines that the meet does not lie within that block.
The subroutine is similar to the main procedure except that it works on the subblock level instead of the block level.


\vspace{0.1in}\hrule\vspace{0.1in}
\noin\textproc{Meet:} Given $x, y \in L$, find $x \meet y$.
\be[(a)]
\item
Initialize an empty set $Z$ to store candidate meets for $x$ and $y$.
\item
\emph{Check principal blocks:} For each principal block $B_i$, find the representative elements $x_i = x \meet h_i$ and $y_i = y \meet h_i$ using ($\cA$).
If $x_i \in B_i$ and $y_i \in B_i$, then use the subroutine \textproc{Meet-In-Block} to either find $x_i \meet y_i$ or determine that $B_i$ does not contain it.
If $x_i \meet y_i$ is found, then add it to $Z$.
\item
\emph{Check residual block:}
If $x$ and $y$ are both in the residual block $B_\res$, then use $\DOWN(x)$ to iterate through every element $z \in \down x \cap B_\res$. Add $z$ to $Z$ whenever $z \le y$.
\item
Using the order-testing operation, determine the maximum element in $Z$ and return it.
If $Z$ is empty, then conclude that the meet of $x$ and $y$ does not exist and return $\NULL$.
\ee

\noin\textproc{Meet-In-Block:} Given $x_i, y_i \in B_i$, either find $x_i \meet y_i \in B_i$ or determine that $x_i \meet y_i \not\in B_i$.
\be[(a)]
\item If $x_i = h_i$ or $y_i = h_i$, then return the smaller of $x_i$ and $y_i$. Otherwise, initialize an empty set $Z_i$ to store candidate meets for $x_i$ and $y_i$ in $B_i$.
\item \emph{Check principal subblocks:} For each principal subblock $S_{i,j}$, find the representative elements $x_{i,j} = x_i \meet g_{i,j}$ and $y_{i,j} = y_i \meet g_{i,j}$ using ($\cC$).
If $x_{i,j}$ and $y_{i,j}$ are both in $S_{i,j}$, then look up
\[z_{i,j} = \begin{cases}x_{i,j} \meet y_{i,j} &\text{if $x_{i,j} \meet y_{i,j} \in S_{i,j}$}\\ \quad\NULL &\text{otherwise}\end{cases}\]
using the appropriate table in ($\cD$). If $z_{i,j} \not= \NULL$ then add it to $Z_i$.

\item \emph{Check residual subblock:} If $x_i$ and $y_i$ are both in the residual subblock $S_{i,\res}$, then use ($\cE$) to iterate through every element $z \in \down x_i \cap S_{i,\res}$. Add $z$ to $Z_i$ whenever $z \le y_i$.
\item 
Using the order-testing operation, determine the largest node in $Z_i$ and return it.
If $Z_i$ is empty, then conclude that $x_i \meet y_i \not\in B_i$ and return $\NULL$.
\ee
\hrule\vspace{0.1in}

\subsection*{Correctness}

We now prove that this algorithm is correct, beginning with the correctness of \textproc{Meet-In-Block}.

\begin{lemma}
	\label{lem:subcorrectness}
	\textproc{Meet-In-Block} returns $x_i \meet y_i$ if it lies in $B_i$ and $\NULL$ otherwise.
\end{lemma}
\begin{proof}
	If $x_i = h_i$ or $y_i = h_i$, then $x_i \meet y_i$ is returned in step (i).
	Otherwise, the correctness of the algorithm relies on two facts.
	\be[{Fact} 1.]
	\item Every element $z \in Z_i$ satisfies $z \le x_i \meet y_i$.
	\item If $x_i \meet y_i$ exists and lies in $B_i$, then it is added to $Z$.
	\ee
	Assuming these hold, step (iv) must correctly answer the query:
	In the case that $x_i \meet y_i \in B_i$, the meet must be added to $Z_i$ and it must the maximum element among all elements in $Z_i$.
	If $x_i \meet y_i \not\in B_i$, then $Z_i$ will be empty by the first fact.
	
	Fact 1 is straightforward.
	Every candidate meet $z$ added to $Z_i$ in step (ii) is $x_{i,j} \meet y_{i,j}$ for some $j \in \{1, \dots, \ell_i\}$, as reported by ($\cD$).
	Since $x_{i,j} \le x_i$ and $y_{i,j} \le y_i$ we have $z \le x_i \meet y_i$.
	When a candidate meet $z$ is added to $Z$ in step (iii) it is because $z \in \down x_i \cap B_\res$ and $z \le y_i$; hence $z \le x_i \meet y_i$.
	
	To prove Fact 2, first suppose that $x_i \meet y_i$ lies in a principal subblock $S_{i,j}$.
	Then $x_i \meet y_i \le g_{i,j}$.
	By the elementary properties of the meet operation,
	\[x_i \meet y_i = x_i \meet y_i \meet g_{i,j} = (x_i \meet g_{i,j}) \meet (y_i \meet g_{i,j}) = x_{i,j} \meet y_{i,j}.\]
	Thus, $x_i \meet y_i$ is added to $Z$ during step (ii) when the subblock $S_{i,j}$ is considered.
	
	Now suppose that $x_i \meet y_i$ lies in the residual subblock $S_{i,\res}$.
	In this case, $x_i$ and $y_i$ must themselves lie in $S_{i,\res}$, for if either one is below any subblock header of $B_i$ then their meet would also be below that same block header.
	Thus, $x_i \meet y_i$ will be added to $Z_i$ in step (3) during which every element of $\down x_i \cap \down y_i \cap S_{i,\res}$ is added to $Z_i$.
	This proves Fact 2.
\end{proof}

\begin{lemma}
	\label{lem:correctness}
	\textproc{Meet} finds $x \meet y$ or correctly concludes that it does not exist.
\end{lemma}
\begin{proof}
	This proof is similar to that of Lemma~\ref{lem:subcorrectness}.
	It relies on the same two facts.
	\be[{Fact} 1.]
	\item Every element $z \in Z$ satisfies $z \le x \meet y$.
	\item If $x \meet y$ exists, then it is added to $Z$.
	\ee
	Assuming these hold, step (4) must correctly answer the query.
	The only significant difference between \textproc{Meet} and \textproc{Meet-In-Block} is the method of finding candidate meets in step (2).
	\textproc{Meet} calls \textproc{Meet-In-Block} to find $x_i \meet y_i$ if it lies in $B_i$ whereas \textproc{Meet-In-Block} uses ($\cD$) to find $x_{i,j} \meet y_{i,j}$ if it lies in $S_{i,j}$.
	By Lemma~\ref{lem:subcorrectness}, \textproc{Meet-In-Block} accurately returns $x_i \meet y_i$ if $x_i \meet y_i \in B_i$ and $\NULL$ otherwise.
	Now Facts 1 and 2 may be proved by the same arguments.
	%
	%
\end{proof}

\subsection*{Time Analysis}

The meet procedure takes $O(n^{3/4})$ time in the worst case.
We first analyze the time for \textproc{Meet-In-Block} applied to a principal block $B_i$.
Step (i) takes constant time.
Step (ii) takes constant time per principal subblock of $B_i$ using ($\cC$) and ($\cD$).
Since each principal subblock has size at least $\sqrt{|B_i|}$, there are at most $|B_i|/\sqrt{|B_i|}=\sqrt{|B_i|}$ principal subblocks; hence the time for step (ii) is $O(\sqrt{|B_i|})$.
Step (iii) performs constant-time order testing on all the elements below $x_i$ in the residual subblock.
By the subblock decomposition method, there are at most $\sqrt{|B_i|}$ such elements.

When step (iv) is reached, $Z_i$ has been populated with at most one element per principal subblock ($\sqrt{|B_i|}$ in total) and at most $\sqrt{|B_i|}$ elements from the residual sublock.
The maximum element in $Z_i$ is found in linear time during this step.
Thus, \textproc{Meet-In-Block} runs in $O(\sqrt{|B_i|})$ time when applied to block $B_i$.

Now the main procedure can be analyzed in a similar fashion.
Step (1) takes constant time.
Step (2) calls \textproc{Meet-In-Block} on every principal block, and hence the total time for step (2) is proportional to $\sum_{i=1}^m \sqrt{|B_i|}$.
By Jensen's inequality, $\sum_{i=1}^m \sqrt{|B_i|}$ is maximized when all the blocks have size $\sqrt n$, since each principal block has size at least $\sqrt n$ and $\sum_{i=1}^m |B_i| \leq n$.
Thus \[\sum_{i=1}^m \sqrt{|B_i|} \leq  \sum_{i=1}^{\sqrt n} n^{1/4} \leq n^{3/4}.\]

As in the analysis of steps (iii) and (iv), steps (3) and (4) take $O(\sqrt n)$ time.
Therefore the time complexity of \textproc{Meet} is $O(n^{3/4})$.

\subsection*{Space Complexity}

The space required to store the nodes, ($\cA$), and ($\cB$) is $O(n^{3/2})$ as in Section~\ref{sec:order}.

Fix $i \in \{1, \dots, m\}$. We show that the parts of ($\cC$), ($\cD$), and ($\cE$) relating to $B_i$ occupy $O(|B_i| \sqrt n)$ space.
Since $\sum_{i=1}^m |B_i| \leq n$, it follows that the entire data structure takes $O(n^{3/2})$ space.

Each array in ($\cC$) requires $O(|B_i|)$ space.
There are at most $\sqrt{|B_i|}$ subblock headers for a total of $O(|B_i|^{3/2})$ space.

The lookup table in ($\cD$) for subblock $S_{i,j}$ takes $O(|S_{i,j}|^2)$ space. Since $\sqrt{|B_i|} \leq |S_{i,j}| \leq \sqrt n$, we have $\sum_{j=1}^{\ell_i}|S_{i,j}|^2 \leq \sqrt n \sum_{j=1}^{\ell_i}|S_{i,j}|$.
Notice $\sum_{j=1}^{\ell_i} |S_{i,j}| \leq |B_i|$ as the subblocks are disjoint subsets of $B_i$.
Therefore the total space occupied by ($\cD$) is $O(|B_i| \sqrt n)$.

The lists stored by ($\cE$) occupy $O(\sqrt{|B_i|})$ space each for a total of $O(|B_i|^{3/2})$ space.
The space charged to block $B_i$ is therefore $O(|B_i|^{3/2} + |B_i| \sqrt n + |B_i|^{3/2}) = O(|B_i| \sqrt n)$.


\subsection*{Preprocessing}
It remains to discuss how to efficiently decompose the lattice and initialize the structures ($\cA$) --- ($\cE$).
Recall that we the lattice is presented initially by its transitive reduction graph.
It is known that the number of edges in the TRG of a lattice is $O(n^{3/2})$~\cite{KL71, Win79}.
We assume that the TRG is stored as a set of $n$ nodes, each with a list of its out-neighbours (nodes that cover it) and a list of in-neighbours (nodes that it covers).
The total space needed for this representation is $O(n^{3/2})$.
The preprocessing takes $O(n^2)$ time and the space usage never exceeds $O(n^{3/2})$.

The first step in preprocessing is to determine the block decomposition.
The same technique will apply to subblock decompositions.
We begin by computing a \emph{linear extension} of the lattice.
A linear extension of a partially-ordered set is an order of the elements $x_1, x_2, \dots, x_n$ such that if $i \leq j$ then $x_j \not\le x_i$. A linear extension may be found by performing a topological sort on the TRG, which can be done in $O(n^{3/2})$ time~\cite{K62}.

We now visit each element of $L$ in the order of this linear extension and determine the size of its downset.
The size of the downset can be computed by a depth-first search beginning with the element and following edges descending the lattice.
This search takes time proportional to the number of edges between elements in the downset.
As soon as this process discovers a fat node $h$ (a node with at least $\sqrt n$ elements in its downset), it can be used as a block header.
Then $h$ and every element of its downset can be deleted from $L$.
The process of computing the sizes of the downsets can continue from the node following $h$ in the linear extension, and the only difference is that the graph searches used to compute the size of each downset must now be restricted to $L \setminus \down h$.
There is no need to recompute the downset size of any node before $h$ in the linear extension because the size of its downset was less than $\sqrt n$ previously and deleting $\down h$ can only reduce this value.
The fat nodes encountered in this way form the block headers of the decomposition.
After every node has been visited, the remaining elements can be put into the residual block.

The time needed for the decomposition depends on the number of edges in each downset.
By Lemma~\ref{lem:blocksAreLattices}, every downset is a partial lattice, and thus a downset with $k$ nodes can have only $O(k^{3/2})$ edges.
For every thin node encountered, the number of edges in the downset is at most $O((\sqrt n)^{3/2}) = O(n^{3/4})$ because it contains less than $\sqrt n$ elements.
Thus, the time needed to visit all the thin nodes is $O(n^{7/4})$.

Whenever a fat node is discovered its downset is removed immediately, and so the edges visited during the DFS are never visited again.
Hence, the time needed to examine all of the block headers is proportional to the number of edges in the whole TRG.
Therefore a block decomposition can be computed in $O(n^{7/4})$ time.

By the same procedure, the subblocks can be computed in $O(\sum_{i=1}^m |B_i|^{7/4})$ time.
Since $\sum_{i=1}^m |B_i| \leq n$, this is at most $O(n^{7/4})$.

With the block and subblock decompositions in hand, data structures ($\cA$) --- ($\cE$) can be initialized. See Appendix~\ref{sec:initialize} for details.

We have now proven the main theorem of this paper.

\begin{theorem}
	\label{thm:main}
	There is a data structure for lattices that requires $O(n^{3/2})$ space, answers order-testing queries in $O(1)$ time, and computes the meet or join of two elements in $O(n^{3/4})$ time. The preprocessing time starting from the transitive reduction graph of the lattice is $O(n^2)$.
\end{theorem}

A straightforward generalization of the data structure allows for a space-time tradeoff.
\begin{corollary}
	\label{cor:generalize}
	For any $c \in [\frac{1}{2}, 1]$, there is a data structure for lattices that requires $O(n^{1+c})$ space and computes the meet or join of two elements in $O(n^{1-c/2})$ time.
	The preprocessing time, starting from the transitive reduction graph of the lattice, is $O(n^2 + n^{1+3c/2})$.
\end{corollary}

\begin{proof}
	The modification is obtained by adjusting the block size of the initial decomposition from $\sqrt n$ to $n^c$. Otherwise, the data structure and methods are identical.
	The time, space, and preprocessing analyses are similar.
\end{proof}
Note that for $c = \frac{1}{2}$, this data structure is precisely that of Theorem~\ref{thm:main}.

\section{Degree-Bounded Extensions}
\label{sec:deg}
Recall that we assume that the lattice is initially represented by its transitive reduction graph (TRG).
Let the \emph{degree} of a lattice node be the number of in-neighbours in the TRG, or equivalently, the number of nodes it covers.
Interestingly, developing methods that handle high-degree nodes efficiently has been the primary obstacle to improving on our data structure.
Indeed, the ``dummy node'' technique of Talamo and Vocca, explained in Appendix~\ref{sec:error}, is effectively used to get around high-degree lattice elements.
We have found that meets and joins can be computed more efficiently as long as the maximum degree of any node in the lattice is not too large.
This is the case for distributive lattices, for example, as $\log_2 n$ is the maximum degree of a node in a distributive lattice\footnote{We leave this as an exercise using Birkhoff's Representation Theorem~\cite{Bir37}.}.
In this section, we explore new data structures for meet and join operations that perform well under this assumption.


Let $d$ be the maximum degree of any node in a partial lattice $L$.
As a convenience, we assume in this section that $L$ has a top element.
The purpose of this assumption is to avoid a lattice with more than $d$ maximal elements;
otherwise we would need to define $d$ as the larger of the maximum degree and the number of maximal elements in the lattice.

This assumption has the effect that the residual block in any block decomposition of $L$ has a top element (unless it is empty).
The only practical difference between the residual block and a principal block is that the residual block may be smaller than the block size of the decomposition.
The results of this section are easier to relate if we assume henceforth that all blocks are principal blocks and each has a block header.
Thus, a block decomposition with block size $k$ creates $m$ blocks $B_1, \dots, B_m$ with block headers $h_1, \dots, h_m$, where $|B_i| \geq k$ for $1 \leq i \leq m-1$.
The number of blocks is at most $\frac{n}{k}+1$.

We begin with a simple data structure that computes joins between elements using a new strategy.
It is more efficient than our earlier method when $d \leq n^{3/4}$.
We then generalize the idea to create a more sophisticated recursive data structure.
It improves on the simple structure for all values of $d$ and works especially well when $d \leq \sqrt n$.
The space usage is $O(n^{3/2})$ for both data structures.
Either one can be used to compute meets as well by inverting the lattice order and rebuilding the data structure, although the value of $d$ may be different in the flipped lattice.

\begin{theorem}
	\label{thm:joinfinding}
	There is a data structure for lattices that requires $O(n^{3/2})$ space and computes the join of two elements in $O(\sqrt n + d)$ time.
\end{theorem}
\begin{proof}
	This data structure uses a block decomposition with block size $k = \sqrt n$ and stores ($\cA$) and ($\cB$) just as in Section~\ref{sec:order}.
	This is everything we need to perform order-testing in constant time.
	However, we now use this information to compute \emph{joins} instead of meets.
	
	Let $B_1, \dots, B_m$ be the blocks of the decomposition with block headers $h_1, \dots, h_m$.
	Further assume that the order $B_1, B_2, \dots, B_m$ reflects the order that the blocks were extracted from $L$ during the decomposition.
	
	Given $x, y \in L$, $x \join y$ may be found as follows.
	\be[(1)]
	\item Use order-testing to compare $x$ and $y$ to every block header.
	Let $i^* \in \{1, \dots, m\}$ be the smallest value for which $x \le h_{i^*}$ and $y \le h_{i^*}$.
	\item
	It must be that $x \join y$ lies in $B_{i^*}$.
	Let $c_1, c_2, \dots, c_t \in B_{i^*}$ be the elements covered by $h_{i^*}$ in $B_{i^*}$\footnote{It is possible that $h_{i^*}$ covers other elements belonging to earlier blocks. These are not included.}.
	Compare $x$ and $y$ to each of these elements using order-testing queries.
	If $x, y \le c_j$ for some $j \in \{1, \dots, t\}$, then proceed to step (3).
	Otherwise, conclude that $x \join y = h_{i^*}$.
	\item
	The join of $x$ and $y$ must lie in the local downset of $c_j$.
	Find $x \join y$ by comparing $x$ and $y$ to every element in $\down c_j \cap B_{i^*}$ and choosing the smallest node $z$ with $x, y \leq z$.
	\ee
	This procedure always finds $x \join y$.
	The purpose of step (1) is to identify the block containing $x \join y$.
    With $i^*$ defined as in the algorithm, observe that $x \join y$ must have been added to $B_{i^*}$ during the decomposition because $x \join y \in \down h_{i^*}$ and $x \join y \not\in \down h_i$ for any $i < {i^*}$.
	This step takes $O(\sqrt n)$ time as $m \leq \sqrt n + 1$.
	
	Once $B_{i^*}$ has been identified, the difficulty lies in finding the join.
	The algorithm checks all of the children $c_1, \dots, c_t$ of $h_{i^*}$ to find an element $c_j$ above $x \join y$.
	This step requires $O(d)$ time as $t \leq d$.	
	If the algorithm succeeds in finding $c_j$ then it compares $x$ and $y$ with all of the elements in the local downset of $c_j$ to determine the join.
	By the thinness property, this step takes only $O(\sqrt n)$ time.
	If no such $c_j$ exists, then $h_{i^*}$ must be the only element in $B_{i^*}$ above both $x$ and $y$.
	Thus, this data structure finds $x \join y$ in $O(\sqrt n + d)$ time.
\end{proof}

\begin{theorem}
	\label{thm:recurse}
	There is a data structure for lattices that requires $O(n^{3/2})$ space and computes the join of two elements in $O(d\frac{\log n}{\log d})$ time.
\end{theorem}

\begin{proof}
	We extend the ideas of Theorem~\ref{thm:joinfinding} using a recursive decomposition of a lattice.
	
	The recursive decomposition works in two stages.
	First, we perform a block decomposition of $L$ using the block size $n/d$.
	This produces up to $d+1$ blocks $B_1, \dots, B_m$.
	
	We decompose each $B_i$ further using a \emph{cover decomposition}.
	If $B_i$ has a block header $h_i$ and $c_1, c_2, \dots, c_t \in B_i$ are the elements covered by $h_i$, then a cover decomposition of $B_i$ is a partition of $B_i$ into the sets
	\[C_{i,j} = (\down c_j \cap B_i) \setminus (\bigcup_{\ell=1}^{j-1} \down c_j) \text{ for } 1 \leq j \leq t.\]
	We call these sets \emph{chunks} to avoid overloading ``block'' and we call $c_j$ the \emph{chunk header} of $C_{i,j}$.
	Unlike a block decomposition, a cover decomposition does not depend on a block size. It is unique up to the ordering of $c_1, \dots, c_t$.
	
	So far, our decomposition produces blocks $\{B_1, \dots, B_m\}$ and chunks $\{C_{i,j} \mid 1 \leq i \leq m, 1 \leq j \leq \deg(h_i)\}$.
	We recursively decompose every chunk $C_{i,j}$ in the same two stages, first by a block decomposition with block size $\frac{|C_{i,j}|}{d}$ and then by a cover decomposition of each of the resulting blocks.
	The recursive decomposition continues in this fashion on any chunk with size at least $2d$.
	
	The recursion induces a tree structure on the set of block headers and chunk headers in the lattice.
	The children of each block header are the chunk headers chosen during its decomposition and vice versa.
	The order of the children of a node corresponds to the order that the blocks or chunks are taken during the decomposition.
	Finally, we create one special node to act as the root of the tree. The children of the root are the block headers of the initial decomposition.
	We call this the \emph{decomposition tree}.
	
	It is easy to see that every lattice element occurs at most once in the tree and that the maximum degree of any tree node is at most $d+1$.
	Less obvious is the fact that the depth of the tree is $O(\frac{\log n}{\log d})$.
	
	To see this, let $c$ be a chunk header, let $h$ be one of its children in the tree, and let $c'$ be a child of $h$.
	Assume $c$ is the header of a chunk $C$, $h$ is the header of a block $B$ contained in $C$, and $c'$ is the header of a chunk $C'$ contained in $B$; see Figure~\ref{fig:decomp}.
	Block $B$ was formed during a block decomposition of $C$ with size $|C|/d$.
	Since $h$ covers $c'$ in $B$, $c'$ must have been a thin node during that decomposition.
	The chunk $C'$ was then created from the local downset of $c'$ in $B$.
	Thus $|C'| \leq |\down c' \cap B| \leq |C|/d$.
	
	\begin{figure}[!ht]
		\begin{center}
			\includegraphics[scale=0.7]{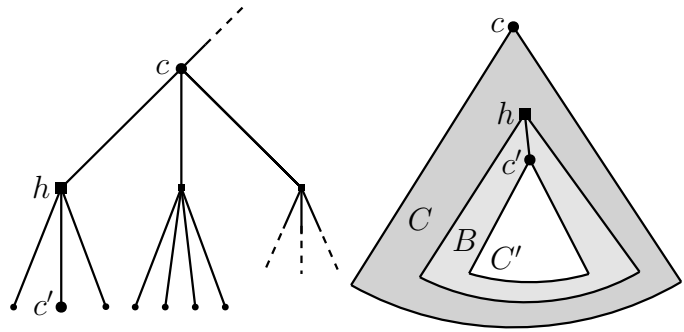}
		\end{center}
		\caption{Three nodes in the decomposition tree and the corresponding chunks and block of the recursive decomposition. The size of $C'$ can be no larger than $|C|/d$.}
		\label{fig:decomp}
	\end{figure}
	
	This implies that the size of chunks decreases by a factor of $d$ between every chunk header and its grandchildren in the decomposition tree.
	After $2\lceil \frac{\log n}{\log d}\rceil$ generations in the decomposition tree, every chunk must have size less than $2d$.
	This proves the claim.
	
	The data structure is now simple to describe.
	We store the decomposition tree and, for each leaf, we store a list of the elements in the chunk of that chunk header.
	Since the chunks represented by leaves are pairwise disjoint, only $O(n)$ space is needed for this structure.
	Additionally, we create and store the order-testing structure of Section~\ref{sec:order}, bringing the total space to $O(n^{3/2})$.

	The join of two elements can be found using a recursive version of the algorithm from Theorem~\ref{thm:joinfinding}.
	Suppose we are given $x, y \in L$ and must determine $x \join y$.
	Through a variable $u$ that represents the node being considered, we recursively traverse the decomposition tree.
	Initially set $u$ equal to the root and proceed as follows.
	
	\subsection*{Base Case:}
	If $u$ is a leaf, then consider the stored list of elements for $u$.
	Find $x \join y$ by comparing $x$ and $y$ to every element in the list and returning the smallest node $z$ with $x, y \leq z$.
	
	\subsection*{Recursive Case:}
	If $u$ is not a leaf, let $v_1, v_2, \dots, v_k$ be the children of $u$ in the decomposition tree, listed in order.
	Use order-testing to compare $x$ and $y$ to each $v_i$.
	If there is no $v_i$ such that $x \le v_i$ and $y \le v_i$, then conclude that $x \join y = u$.
	Otherwise, let $i^* \in \{1, \dots, k\}$ be the smallest value for which $x \le v_{i^*}$ and $y \le v_{i^*}$.
	Recurse on $v_{i^*}$.
	
	This procedure spends $O(d)$ time on each node.
	In the base case, the list stored for $u$ has length $O(d)$ and the join can be found in this list in linear time.
	The recursive case takes $O(d)$ time as well since the maximum degree of the decomposition tree is at most $d+1$.
	As the depth of the tree is $O(\frac{\log n}{\log d})$, the total time of this procedure is $O(d \frac{\log n}{\log d})$.
	
	Correctness is a consequence of the fact that $x \join y$ lies in the first block of each block decomposition whose header is above both $x$ and $y$.
	The same fact holds for the chunks in a cover decomposition.
	Thus, each time $i^*$ is chosen in the recursive case, it must be that $x \join y$ lies in the block or chunk for $v_{i^*}$.
\end{proof}

\section{Conclusions}
\label{sec:conclusion}

We have presented a data structure to represent lattices in $O(n^{3/2})$ words of space, which is within a $\Theta(\log n)$ factor of optimal. It answers order queries in constant time and meet or join queries in $O(n^{3/4})$ time.
This work is intended to replace the earlier solution to this problem which was incorrect; see Appendix~\ref{sec:error} for a discussion of the error.
Our degree-bounded data structure uses $O(n^{3/2})$ space and answers meet or join queries in $O(d \frac{\log n}{\log d})$ time. For some low-degree lattices, this structure improves dramatically on our subblock-based approach.
Ours are the only data structures known to us that uses less than the trivial $O(n^2)$ space.

We wonder what can be done to improve on our results.
The time to answer meet and join queries may yet be reduced, perhaps to the $O(\sqrt n)$ bound claimed by~\cite{TV99}.
Another natural question is whether the space of the representation can be reduced to the theoretical minimum of $\Theta(n^{3/2})$ bits.

\bibliography{LAT}
\newpage
\appendix

\section{Initializing the Data Structure}
\label{sec:initialize}

We now show how ($\cA$), ($\cB$), ($\cC$), ($\cD$), and ($\cE$) can be constructed in $O(n^2)$ time.
We assume that we have access to the TRG of the lattice and that the block and subblock decompositions have already been computed.

\begin{itemize}
	\item[($\cA$)] Some care is required to construct ($\cA$) efficiently.
	Let $x_1, \dots, x_n$ be a linear extension of $L$.
	Consider a principal block $B_i$ with block header $h_i$.
	To find $z \meet h_i$ for each $z \in L$, we do the following.
	\be
		\item Initialize an array of length $n$ to store the meet of $h_i$ with each element and populate the array with $\NULL$ in every entry.
		
		\item Perform a DFS to find $\down h_i$ in $L$. Note that $\down h_i$ may be considerably larger than $B_i$.
		Put the elements of $\down h_i$ (note that this includes $h_i$) into a linear extension $y_1, \dots, y_k$ by restricting the linear extension of $L$ to these elements.
		
		\item Traverse the nodes in reverse order of this extension (beginning with $y_k$ and ending at $y_1$).
		For each node $y_j$, perform a DFS on the \emph{upset} of that node in the full lattice.
		For every node $z$ visited during the DFS for node $y_j$, record that $z \meet h_i = y_j$ in the array, and then mark $z$ so that it will not be visited by later graph searches.
		After all the nodes in $\down h_i$ have been processed, restore the lattice by unmarking all nodes.
	\ee
	By this method, the entry for $z \meet h_i$ in the array is recorded to be the last element in the linear extension of $\down h_i$ that is below $z$.
	This must be the correct node because it is below both $z$ and $h_i$, and every other element below $z$ and $h_i$ occurs earlier in the linear extension.
	Whenever $z \meet h_i$ does not exist in the lattice, the array entry for $z \meet h_i$ is the default value $\NULL$.
	
	The time for this procedure is bounded by the number of edges in the TRG for $L$ because no node is visited more than once over all of the graph searches.
	Recall that the number of edges in the TRG is $O(n^{3/2})$.
	Summing over all block headers, the total time to create ($\cA$) is at most $O(n^{3/2} \sqrt n) = O(n^2)$.

	\item[($\cB$)] This can be computed by performing a DFS on the local downset of each node and adding the elements visited to a dictionary for that node.
	
	Initializing and populating the space-efficient dictionary of~\cite{BM99} takes time linear in the number of dictionary entries.
	Excluding the block headers, the local downsets have at most $\sqrt n$ nodes and $O(n^{3/4})$ edges; hence the time spent on all non-block headers is at most $O(n^{7/4})$.
	The local downsets of the block headers are all disjoint, so the total time required is $O(n^{7/4})$.
	
	\item[($\cC$)] Use the same method for ($\cA$) restricted to each block to compute ($\cC$). The total time is $O(\sum_{i=1}^m |B_i|^2)$, which is no larger than $O(n^2)$.
	
	\item[($\cD$)]  The method of ($\cA$) can also be used to compute ($\cD$). For each element $z$ in a principal subblock $S_{i, j}$, find the meet of $z$ with every other element in the subblock in $O(|S_{i,j}|^{3/2})$ time, where $z$ plays the role of $h_i$ in the method for ($\cA$). It takes $O(|S_{i,j}|^{5/2})$ time to do this for every element in a single subblock and the total time is proportional to
	\[
	\sum_{i=1}^m \sum_{j=1}^{\ell_i} |S_{i,j}|^{5/2} \leq  \sum_{i=1}^m \sum_{j=1}^{\ell_i} |S_{i,j}|(\sqrt{n})^{3/2} \leq n^{7/4}.
	\]
	The first inequality uses the fact that each subblock has size at most $\sqrt n$.
	The second inequality holds because the subblocks are disjoint.
	
	\item[($\cE$)] Each linked list can be constructed by performing a DFS on the downset of each element in a residual subblock. This takes $O(n^{7/4})$ time as in the analysis for ($\cB$).
\end{itemize}

\section{Correcting Earlier Work}
\label{sec:error}
As stated in the introduction, this paper relies on ideas from the lattice data structure of~\cite{TV94, TV97, TV99}.
These papers contain a mistake that we believe is not easily repaired.
The purpose of this section is to summarize their techniques, explain where the error occurs, and argue that it cannot be fixed by a minor modification. We urge the interested reader to consult~\cite{TV99} to confirm this analysis.

We restate their algorithm in the language of this paper.
In the interest of a clear and concise explanation, we do not rebuild all the machinery of their work.
In particular, we ignore their \emph{double-tree} structure and we only consider blocks made from downsets (in their papers, blocks may be built from upsets or downsets). In our observation, the double-tree structure is necessary only as a null/non-null value check for order testing and meet/join queries (thus a simple dictionary suffices); further, while we have concerns about using both upsets and downsets for blocks, using downsets alone avoids such issues and still satisfies the requirements in their papers (Lemma 4.1 in~\cite{TV99}). We take these liberties for the purpose
of quickly coming to the relevant issue.
Readers will need to confirm for themselves
that our explanation is fundamentally accurate.


Their method relies on a lattice decomposition to build the data structure, and our block decomposition is similar to the basic version of the decomposition described in their papers.
Note that what we call ``blocks'' are called ``ideals'' in~\cite{TV97} and ``clusters'' in~\cite{TV99}.
They do not decompose the lattice at a second level like our subblock decompositions.
The error is introduced in the extended version of their lattice decomposition, which we now describe.

The intuition behind their data structure is that everything would be easier if every block had size $\Theta(\sqrt n)$, say between $\sqrt n$ and $2\sqrt n$.
If this were the case, then we could afford to explicitly store the meet/join and reachability property
between every pair of elements from the same block,
as this would use roughly $\sum_{i=1}^{\sqrt n} (\sqrt n)^2 = O(n^{3/2})$ space.
This would allow the meet of two elements from the same block to be found in constant time by a simple table lookup.
In terms of our meet-finding algorithm from Section~\ref{sec:meet}, this would reduce the time for \textproc{Meet-In-Block} to a constant and the time for \textproc{Meet} to $O(\sqrt n)$.

\subsection*{Dummy Nodes}

A block decomposition by itself cannot guarantee anything about the sizes of the blocks except that each is at least $\sqrt n$.
They attempt to simulate blocks of size $\sqrt n$ by modifying the transitive reduction graph (TRG) of the lattice, creating ``dummy nodes'' with downsets of size $\Theta(\sqrt n)$ to act as block headers when none exist naturally.




Dummy nodes are introduced as follows.
Suppose a block $B$ is created that has more than $2\sqrt n$ elements.
Assume that the block header has children $c_1, c_2, \dots, c_t$ in the TRG.
Consider the sequence 
\[|\down c_1 \cap B|, |(\down c_1 \cup \down c_2) \cap B|, \dots, |(\down c_1 \cup \cdots \cup \down c_t) \cap B|.\]
As each of the children is a thin element (its local downset has size less than $\sqrt n$), the difference between adjacent numbers in this sequence is less than $\sqrt n$.
Thus, there is some $i \in \{1, \dots, k\}$ such that \[\sqrt n \leq |(\down c_1 \cup \cdots \cup \down c_i) \cap B| \leq 2 \sqrt n.\]

The children $c_1, \dots, c_i$ may be grouped together and the set $(\down c_1 \cup \cdots \cup \down c_i) \cap B$ may be considered as an \emph{artificial} block having size $\Theta(\sqrt n)$.
By removing this artificial block and iterating on the remaining children, $B$ is partitioned into a collection of artificial blocks with sizes between $\sqrt n$ and $2 \sqrt n$ (except that there may be one smaller block at the end).
The only difference between these artificial blocks and ordinary principal blocks is that they lack a block header.

To remedy this, a dummy node is introduced at the top of each artificial block.
That is, a new element $d$ is created and inserted into the TRG with $c_1, \dots, c_i$ as its in-neighbours and the block header of $B$ as its only out-neighbour.

\begin{figure}[!ht]
	\begin{center}
		\includegraphics[scale=0.7]{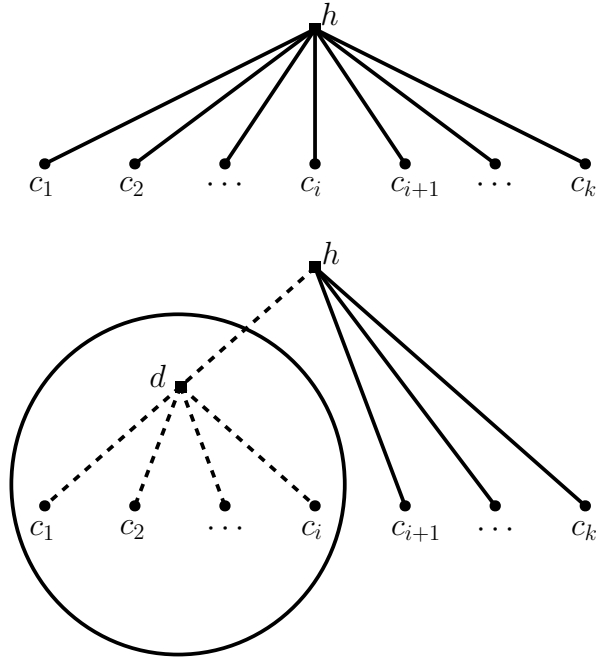}
	\end{center}
	\caption{Dummy nodes are inserted between the block header and its children to simulate blocks of size $\Theta(\sqrt n)$.}
	\label{fig:dummy}
\end{figure}

Talamo and Vocca rely on the fact that the graph still represents a partial lattice after adding dummy nodes in this way. They state on page 1794 of~\cite{TV99}:
\begin{quote}
``By construction, the dag obtained by adding dummy vertices still satisfies the lattice property.''
\end{quote}
Unfortunately, this claim is not true in many cases.
Consider the stripped-down example in Figure~\ref{fig:error}.
The lattice on the left is changed to the graph on the right by introducing a dummy node as described.
However, the graph on the right fails the lattice property because the join of $x$ and $y$ is not well-defined:
Both $c_3$ and $d$ are minimal among elements in $\up x \cap \up y$.
Symmetrically, the meet of $c_3$ and $d$ is not well-defined either.
In this case, adding $d$ broke the lattice property.

Although the example is on a very small lattice, it scales easily to any size.
Any number of nodes could be added to the original lattice so that $|\down c_1 \cup \down c_2| \in [\sqrt n, 2\sqrt n]$.
The dummy node added in this case would still violate the lattice property.

\begin{figure}[!ht]
	\begin{center}
		\includegraphics[scale=0.7]{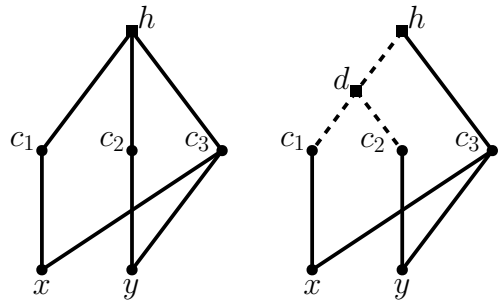}
	\end{center}
	\caption{Inserting a dummy node breaks the lattice property.}
	\label{fig:error}
\end{figure}

This detail is easy to overlook, especially since $\down d \cap B$ is necessarily a partial lattice.
However, the lattice property may fail in the larger structure when dummy nodes are added. This fundamentally impacts the correctness of their approach.

\subsection*{Impact Of The Error}

With dummy nodes, it is no longer true that every element has a unique representative in each block. Talamo and Vocca use the following text on page 1789 of~\cite{TV99}, ``given an external vertex $v$, the pair $(v, Clus(c))$ univocally identifies a vertex $u \in Clus(c)$ representing either the $\texttt{LUB}(Clus^+(c) \cap Clus^+(v))$ or the $\texttt{GLB}(Clus^-(c)\cap Clus^-(v))$.'' In the language of our paper, the claim is that for every block header $h$, any external element $v$ must have a unique representative $x \meet h$. Consider again the example in Figure~\ref{fig:error} with $d$ as the block header of its downset. The external element $c_3$ does not have a unique representative in the block headed by $d$, since the meet of $d$ and $c_3$ is now undefined.

In~\cite{TV99}, this breaks Lemma 3.1 when $c$ is a dummy node, which in turn breaks Lemma 3.3 and implies their data structure $C$ on page 1792 of~\cite{TV99} would need to keep multiple entries for an element-cluster pair in order to guarantee correctness of the reachability algorithm described below it. We see no reason why the number of such representatives stored per element should be small, nor that the total number of representatives stored should be small, which undermines both the proposed query and space complexities.

The same issue arises in Talamo and Vocca's meet and join algorithms. The algorithm given relies on the unique representative of an element with a block, and without it, neither the $O(n\sqrt{n})$ space bound nor the $O(\sqrt{n})$ time bound on meet or join operations follow in Proposition 6.4 of~\cite{TV99}.

Further, if dummy nodes are avoided altogether, the space bound can be $\Omega(n^2)$, as explained on page 1793 of~\cite{TV99}.

It has been suggested to us that the issues may be avoided if the dummy nodes are not considered as actual nodes of the lattice itself, but instead as a construct to group small clusters together for a counting reason. That is, the claim is that an $O(n \sqrt{n})$-space $O(1)$-time order-testing structure can be made without tangibly introducing dummy nodes. As we have shown in this paper, this is indeed true. However, let us emphasize that the work of Talamo and Vocca does not achieve this. It describes a very different technique that crucially relies on the unique representative property remaining true after grouping clusters using dummy nodes, which does not hold in general regardless of whether dummy nodes are actually inserted into the graph or just used as a conceptual tool. With their techniques, we see no way to achieve their claimed $O(\sqrt{n})$ time meet/join algorithm without their erroneous dummy nodes. We give evidence in the following section as to why this might be infeasible.

\subsection*{Can It Be Fixed?}

It is natural to search for a small change to the dummy node method that will fix this issue,
allowing us to effectively perform a block decomposition where every principal block has size $\Theta(\sqrt n)$.
It is especially tempting to do so because it could reduce the time for meet and join operations from $O(n^{3/4})$ to $O(\sqrt n)$, as is claimed in~\cite{TV99}.
The dummy node technique also seems like a reasonable approach to handling high-degree lattice nodes, which have often been an obstacle to the approaches we have considered.

There is good reason to expect that this is not possible, relying on some small assumptions.
Suppose that there were a correct method of creating artificial principal blocks and that the method still works when we increase the block size from $\sqrt n$ to $n^{2/3}$.
That is, suppose that we can reliably decompose any lattice into $\Theta(n^{1/3})$ blocks of size $\Theta(n^{2/3})$ (and perhaps some $O(n^{1/3})$ smaller blocks). Note that the structure of the lattice has no impact on the ability to apply this method, thus we assume it applicable to all lattices.

There is the remaining issue of the residual block, however this is not a major difficulty.
By adding a top element to the partial lattice (as in a complete lattice), we can treat the residual block in the same fashion as a principal block using the new top element as its block header.

Since the number of lattices on $k$ elements is $2^{\Theta(k^{3/2})}$, it is possible to uniquely identify any such lattice using only $\Theta(k^{3/2})$ bits.
Thus, each block of size $\Theta(n^{2/3})$ can be encoded in $\Theta(n)$ bits, and all of the blocks in the decomposition can be encoded in $\Theta(n^{4/3})$ bits.
The order between any pair of elements in the same block can be tested, however inefficiently, using the encoding for that block.
As well, since there are only $\Theta(n^{1/3})$ block headers, all of the meets between a lattice element and a block header can be stored in $\Theta(n^{4/3})$ space.
In other words, we can simulate both ($\cA$) and ($\cB$) in only $O(n^{4/3})$ space.

This information is sufficient to perform order-testing between any pair of elements, and thus it uniquely determines the lattice.
Lattices do not permit such a small representation; this would violate the $\Theta(n^{3/2})$-bit lower bound.
This strongly suggests that artificial blocks cannot be simulated without sacrificing the unique representative property, which is essential to the data structure.

\end{document}